%% file: main.tex
\documentclass[journal]{IEEEtran}

\usepackage{amsmath,amssymb,amsfonts}
\usepackage{bm}
\usepackage{booktabs}
\usepackage{graphicx}
\usepackage{tikz}
\usetikzlibrary{arrows.meta,positioning,fit,backgrounds}
\usepackage{cite}
\usepackage{algorithm}
\usepackage{algpseudocode}
\usepackage{array}
\usepackage{url}

\newtheorem{proposition}{Proposition}

\newcommand{\ind}{\mathbb{I}}
\newcommand{\cset}{\mathcal{C}}
\newcommand{\fset}{\mathcal{F}}
\newcommand{\rset}{\mathcal{R}}
\newcommand{\pos}[1]{\left[#1\right]_{+}}

\begin{document}

\title{SPLIT-Q: A Scalable Sequential Quantum Computing Framework for Coherent Controlled Islanding}

\author{Yuqi~Jiang,~\IEEEmembership{Graduate Student Member,~IEEE},
        Zhiding~Liang,~\IEEEmembership{Member,~IEEE},
        Qiang~Guan,~\IEEEmembership{Senior Member,~IEEE},
        Yan~Li,~\IEEEmembership{Senior Member,~IEEE},
        and~Ganesh~Kumar~Venayagamoorthy,~\IEEEmembership{Fellow,~IEEE}%
\thanks{This work is supported by the Office of Naval Research under award
N00014-22-1-2504 and the National Science Foundation under awards OAC-2417773
and ECCS-2413237/2413238.}%
\thanks{Y. Jiang and Y. Li are with the Department of Electrical Engineering,
The Pennsylvania State University, University Park, PA 16802, USA.}%
\thanks{Z. Liang is with the Department of Computer Science,
Rensselaer Polytechnic Institute, Troy, NY 12180, USA.}%
\thanks{Q. Guan is with the Department of Computer Science,
Kent State University, Kent, OH 44242, USA.}%
\thanks{G. K. Venayagamoorthy is with the Real-Time Power and Intelligent
Systems Laboratory, Holcombe Department of Electrical and Computer Engineering,
Clemson University, Clemson, SC 29634, USA, and with the University of Pretoria,
Pretoria, South Africa.}%
\thanks{Corresponding author: Yan Li (e-mail: yql5925@psu.edu).}}

\pagestyle{plain}

\maketitle

\begin{abstract}
Growing integration of distributed energy resources increases power-system
variability and uncertainty. During disturbances, these effects can intensify
generation--load imbalances and cascading failures.
Controlled islanding limits
their propagation by partitioning a compromised grid into connected,
electrically sustainable islands. However, classical methods face rapidly
growing computational costs as network size and island count increase. Quantum
optimization offers an alternative for exploring this combinatorial partition
space. Yet monolithic quantum formulations encode all assignment decisions in
one circuit, causing qubit demand and circuit complexity to scale with network
size. In this study, a qubit-bounded sequential distributed quantum
approximate optimization algorithm (QAOA) framework is proposed to tackle
coherent controlled islanding under limited quantum resources. It formulates
the optimization as boundary-conditioned regional quadratic unconstrained
binary optimization (QUBO) subproblems that are solved sequentially within a
fixed qubit budget. Thus, circuit width remains independent of network size,
with aggregate quantum workload scaling linearly on bounded-degree networks.
Evaluation covers eleven IEEE systems from 9 to 300 buses using IBM quantum
computing resources, with Gurobi and monolithic QAOA as references. Across all
systems, the framework recovers feasible Gurobi-optimal partitions under
noise, confirming the resilience of its solution quality. The results further
show that the proposed method
substantially reduces quantum-resource demand and circuit complexity relative
to monolithic QAOA, allowing large islanding problems to be addressed within
current hardware limits.
The proposed framework provides a feasible and scalable pathway for quantum
optimization in large-scale power systems.
\end{abstract}


\input{introduction}
\input{problem_formulation}
\input{proposed_method}
\input{numerical_examples}
\input{conclusion}

\bibliographystyle{IEEEtran}
\bibliography{References}

\end{document}

%% file: introduction.tex
\section{Introduction}
\label{sec:introduction}

The growing penetration of distributed energy resources, particularly
weather-dependent wind and solar generation, increases operating uncertainty and
can aggravate generation--load imbalances during disturbances
\cite{chen2022reinforcement,ren2022interpretable}. These conditions increase
the likelihood that local outages develop into cascading failures and widespread
blackouts, motivating effective preventive operations
\cite{ren2022interpretable,vaiman2012risk}.

Controlled islanding is a key emergency operation that contains cascading
failures by opening selected transmission lines to separate a compromised network
into connected, electrically sustainable islands. It preserves critical
generation and load while supporting subsequent restoration
\cite{kyriacou2018controlled}. The resulting optimization seeks a
low-disruption partition into viable, self-sufficient islands, but its candidate
space grows exponentially with network size and each partition must satisfy
multiple coupled operational requirements.

Classical optimization has provided the principal computational foundation for
controlled islanding, yielding methods that balance physical fidelity, search
efficiency, and scalability.
Slow-coherency-based partitioning groups generators by electromechanical
dynamics before assigning the remaining buses
\cite{you2004slow,yang2006slowcoherency}, while ordered-binary-decision-diagram
and decision-tree methods improve candidate-splitting search
\cite{ahmed2003scheme,sun2003splitting,senroy2006decisiontree}. Spectral
graph-partitioning formulations recast islanding as a weighted-cut problem
\cite{ding2013twostep,kyriacou2018controlled}, and mixed-integer formulations
directly optimize operational and stability constraints
\cite{patsakis2019strong,teymouri2019toward,ma2023controlled}. Despite these
advances, reliance on heuristic search, iterative relaxation, or
problem-specific decomposition limits scalability as network size and the
number of islands increase.

Quantum computing offers an alternative paradigm for combinatorial
optimization. By exploiting superposition and entanglement, quantum processors
can, in principle, encode and manipulate exponentially large solution spaces
with a linear number of qubits, motivating their use for classically difficult
problems. Variational gate-based algorithms and quantum annealing are leading
near-term strategies because they target discrete partitioning and assignment
structures under current hardware limits on qubit count and noise.

These capabilities have motivated a growing body of work applying quantum
computing across a wide range of domains, including quantum chemistry
simulation \cite{peruzzo2014variational}, machine learning
\cite{havlicek2019supervised}, financial portfolio and risk optimization
\cite{egger2020quantum}, sensor placement and contingency analysis
\cite{jiang2025pmu,jiang2025contingency}, combinatorial scheduling
\cite{nikmehr2022quantum,koretsky2021adapting}, and computational-advantage
demonstrations on near-term hardware \cite{arute2019quantum,zhong2020quantum}.
Collectively, these results show quantum optimization evolving from theory into
a practically validated tool across diverse domains.

Within this broader effort, quantum algorithms have recently been applied to
controlled islanding specifically. Quantum-annealing formulations have been used
to partition power grids by mapping the assignment problem onto
hardware-embedded qubits \cite{hartmann2025quantum}. Gate-based approaches have
also been explored: prior work has proposed a resource-efficient hybrid quantum
approximate optimization algorithm (QAOA) framework for physics-constrained
islanding \cite{jiang2026regrid} and a
physics-constrained QAOA formulation that further reduces the qubit footprint
required to encode bus-assignment variables \cite{jiang2026pace}.

Classical-postprocessing feeding strategies \cite{jiang2026regrid} and
Lagrangian-based, physics-constrained encodings \cite{jiang2026pace} reduce
qubit requirements and improve solution quality. Nevertheless, monolithic
QAOA still couples all bus assignments in one register, whose width grows with
network size and island count. Under current noisy intermediate-scale
quantum (NISQ) hardware, qubits remain a scarce and error-prone resource, so
this growth limits the instances that can be represented under a fixed qubit
budget regardless of how efficiently each qubit is used. This tension motivates
decomposing the monolithic search itself into a sequence of smaller,
bounded-width subproblems rather than continuing to compress a single
fixed-width encoding. Accordingly, this work advances scalable quantum
controlled islanding through sequential regional optimization and full-network
coordination, with the following contributions:
\begin{itemize}
\item This work proposes a qubit-bounded sequential regional QAOA framework for
coherent controlled islanding that decouples the simultaneous qubit requirement
from network size, providing a scalable route to large islanding problems on
fixed-capacity quantum hardware.
\item An analytical complexity characterization establishing a
network-size-independent circuit width, an aggregate quantum workload growing
linearly on bounded-degree networks, and full-network feasibility evaluation as
the dominant large-scale cost on sparse grids.
\item Cross-platform validation on eleven IEEE systems from 9 to 300 buses,
where feasible Gurobi-optimal cuts are preserved across ideal, calibrated-noise,
and IBM hardware sampling, evidencing noise resilience. The attendant reductions
in qubit demand and compiled circuit complexity relative to monolithic QAOA make
large controlled-islanding instances executable on hardware that their
monolithic formulations exceed.
\end{itemize}


%% file: problem_formulation.tex
\section{Problem Formulation}
\label{sec:formulation}

Intentional controlled islanding separates a stressed transmission network into
coherent, connected subsystems while limiting the power transfer interrupted by
opened branches. This requirement leads to a constrained weighted graph-partitioning
problem.

\subsection{Network Model and Island-Assignment Variables}

The power grid is an undirected weighted graph
$\mathcal{G}=(\mathcal{V},\mathcal{E})$, assumed connected, with bus set
$\mathcal{V}$, branch set $\mathcal{E}$, and $N=|\mathcal{V}|$. The network is
partitioned into $K\ge2$ islands indexed by $\mathcal{K}=\{0,\ldots,K-1\}$.
Each branch $\{i,j\}\in\mathcal{E}$ is assigned a disruption weight based on
the average magnitude of its pre-islanding
directional active-power transfers $P_{ij}^{0}$ and $P_{ji}^{0}$ (in the units
of the input power-flow data)~\cite{jiang2026pace}:
\begin{equation}
w_{ij}=\frac{|P_{ij}^{0}|+|P_{ji}^{0}|}{2},
\qquad \{i,j\}\in\mathcal{E}.
\label{eq:branch_weight}
\end{equation}
This symmetric weight measures the power transfer interrupted by opening the
branch, so cutting a heavily loaded interface incurs a larger objective cost.
Consequently, the disruption objective retains the unit of the supplied
branch-flow data and is reported in MW for all numerical test cases.

Generator and load buses are denoted
$\mathcal{V}_{\mathrm{G}},\mathcal{V}_{\mathrm{L}}\subseteq\mathcal{V}$. For every bus $i\in\mathcal{V}$ and island
$g\in\mathcal{K}$, the binary membership variable
\begin{equation}
y_{ig} =
\begin{cases}
1, & \text{if bus $i$ belongs to island $g$},\\
0, & \text{otherwise}
\end{cases}.
\label{eq:membership}
\end{equation}
Collecting these variables gives
$\bm y=(y_{ig})_{i\in\mathcal V,g\in\mathcal K}$. A complete label vector
$\bm{x}\in\mathcal{K}^{N}$ induces this membership through
$y_{ig}(\bm{x})=\ind\{x_i=g\}$, where $\ind\{\cdot\}$ denotes the indicator
function. This physical membership model is distinct from
the quantum assignment encoding introduced in Section~\ref{sec:method}.

A fixed collection of coherent generator groups
$\mathcal{H}=\{\mathcal{H}_0,\ldots,\mathcal{H}_{K-1}\}$ satisfies
$\varnothing\neq\mathcal{H}_g\subseteq\mathcal{V}_{\mathrm{G}}$ and
$\mathcal{H}_a\cap\mathcal{H}_b=\varnothing$ for $a\neq b$. Because island
labels are interchangeable, this indexing carries no physical significance
until Section~\ref{sec:method} associates $\mathcal{H}_g$ with island $g$
without excluding any physically distinct partition. These assignments
determine which branches a partition opens and therefore its weighted
disruption cost.

\subsection{Weighted Network-Disruption Objective}

An effective partition preserves important transmission corridors whenever the
operational constraints permit. Its disruption cost is the total weight of the
branches opened by the partition,
$C(\bm{x})=\sum_{\{i,j\}\in\mathcal{E}} w_{ij}\ind\{x_i\neq x_j\}$.
Under a valid membership assignment,
\begin{equation}
\ind\{x_i\neq x_j\}
=1-\sum_{g\in\mathcal{K}}y_{ig}y_{jg}.
\label{eq:cut_identity}
\end{equation}
Consequently, the equivalent binary quadratic objective is
\begin{equation}
C(\bm{y}) =
\sum_{\{i,j\}\in\mathcal{E}}w_{ij}
\left(1-\sum_{g\in\mathcal{K}}y_{ig}y_{jg}\right).
\label{eq:qubo_cut}
\end{equation}
The total branch weight is fixed by the pre-islanding operating point. Within the
feasible set, minimizing $C(\bm{y})$ therefore favors partitions that preserve
high-weight transmission interfaces inside the islands. This benefit requires
every island to remain structurally coherent and adequately resourced.

\subsection{Controlled-Islanding Constraints}

A viable islanding solution must partition the entire network into
operationally and structurally sound subsystems. The following constraints enforce these requirements:

\subsubsection{Unique bus assignment}

Every bus is assigned to exactly one island, so
$\sum_{g\in\mathcal{K}}y_{ig}=1$ for every $i\in\mathcal{V}$.

\subsubsection{Minimum island size and resource presence}

Let $M_{\mathrm{V}}$, $M_{\mathrm{G}}$, and $M_{\mathrm{L}}$ denote the
required minimum numbers of buses, generator buses, and load buses in every
island:
\begin{subequations}
\label{eq:cardinality_constraints}
\begin{align}
\sum_{i\in\mathcal{V}}y_{ig} &\ge M_{\mathrm{V}},
&&\forall g\in\mathcal{K}, \label{eq:min_buses}\\
\sum_{i\in\mathcal{V}_{\mathrm{G}}}y_{ig} &\ge M_{\mathrm{G}},
&&\forall g\in\mathcal{K}, \label{eq:min_generators}\\
\sum_{i\in\mathcal{V}_{\mathrm{L}}}y_{ig} &\ge M_{\mathrm{L}},
&&\forall g\in\mathcal{K}. \label{eq:min_loads}
\end{align}
\end{subequations}
These bounds prevent empty or resource-deficient islands and can be raised
above one to require additional redundancy.

\subsubsection{Generator coherency and group separation}

The coherent groups in $\mathcal H$ become fixed island anchors and all generators in the same coherent group remain
together,
\begin{equation}
y_{ig}=y_{jg},
\quad
\forall i,j\in\mathcal{H}_a,\
\forall a\in\mathcal{K},\
\forall g\in\mathcal{K},
\label{eq:coherency}
\end{equation}
while distinct coherent groups are separated into distinct islands:
\begin{equation}
\sum_{g\in\mathcal{K}}y_{ig}y_{jg}=0,
\quad
\substack{\forall i\in\mathcal{H}_a,\ j\in\mathcal{H}_b,\\
a,b\in\mathcal{K},\ a\neq b.}
\label{eq:separation}
\end{equation}
Because the number of coherent groups equals the number of islands,
\eqref{eq:coherency}--\eqref{eq:separation} place exactly one coherent group in
each island.

\subsubsection{Topological connectivity}

For $\mathcal B\subseteq\mathcal V$, let $\mathcal G[\mathcal B]$ denote the
subgraph of $\mathcal G$ induced by $\mathcal B$. For each $g\in\mathcal{K}$, let
\begin{equation}
\mathcal{B}_g(\bm{y})=\{i\in\mathcal{V}:y_{ig}=1\},
\qquad
\mathcal{G}_g(\bm y)=\mathcal{G}[\mathcal{B}_g(\bm y)]
\label{eq:induced_subgraph}
\end{equation}
denote its bus set and induced subgraph. Each island must form a contiguous
electrical subsystem, so $\mathcal{G}_g(\bm y)$ must be connected for every
$g\in\mathcal{K}$.
Letting $\kappa_g(\bm{y})$ denote the number of connected components of
$\mathcal{G}_g(\bm y)$, with $\kappa_g=0$ for an empty island, define the
full-network connectivity violation as
$r^{\mathrm{con}}_g(\bm{y})=\ind\{\kappa_g(\bm{y})\neq 1\}$.
This condition is checked after measurement: regional labels are merged into a
complete network assignment, and a depth-first search (DFS) on every induced
subgraph $\mathcal{G}_g(\bm y)$ must visit every member of every island for the
assignment to be feasible.

Together with the disruption objective, these feasibility conditions
define the original constrained optimization problem.

\subsection{Compact Constrained Optimization Model}

Controlled islanding is thus a constrained partitioning problem in which reduced
disruption cannot be traded for physical inadmissibility. Combining the preceding
objective and constraints gives
\begin{subequations}
\label{eq:full_model}
\begin{align}
C^\star = \min_{\bm{y}}\quad
&\sum_{\{i,j\}\in\mathcal{E}}w_{ij}
\left(1-\sum_{g\in\mathcal{K}}y_{ig}y_{jg}\right)
\label{eq:full_obj}\\
\text{s.t.}\quad
&\sum_{g\in\mathcal K}y_{ig}=1,\quad \forall i\in\mathcal V,\notag\\
&\eqref{eq:cardinality_constraints},\
\eqref{eq:coherency},\ \eqref{eq:separation},\
r^{\mathrm{con}}_g(\bm y)=0,\quad \forall g\in\mathcal K,
\label{eq:full_constraints}\\
&y_{ig}\in\{0,1\},
\qquad \forall i\in\mathcal{V},\ g\in\mathcal{K}.
\label{eq:binary_domain}
\end{align}
\end{subequations}
The problem is assumed to have at least one feasible assignment. Let
\begin{equation}
\fset=\left\{\bm{x}\in\mathcal{K}^{N}:\bm{y}(\bm{x})
\text{ satisfies \eqref{eq:full_constraints}}\right\}
\label{eq:label_feasible_set}
\end{equation}
denote the feasible set of complete label vectors. The quantum subproblems
developed next generate regional candidates, but feasibility and final solution
quality are always defined with respect to $\fset$ and the physical objective
$C(\bm{x})$.

When a regional candidate pool contains no feasible point, binary feasibility alone
does not distinguish the degree of violation. Residuals provide this secondary
ordering without changing the original feasible objective.

\subsection{Residuals and Feasibility-First Candidate Ranking}

Regional optimization can temporarily produce candidates outside the physical
feasible set. When no feasible candidate is available, smaller aggregate
violations provide a principled secondary ordering. With
$[a]_+=\max\{a,0\}$, the global inequality residuals are defined as
\begin{subequations}
\label{eq:residuals}
\begin{align}
r^{\mathrm{V}}_g(\bm{y})
&=\pos{M_{\mathrm{V}}-\sum_{i\in\mathcal{V}}y_{ig}},\\
r^{\mathrm{G}}_g(\bm{y})
&=\pos{M_{\mathrm{G}}-\sum_{i\in\mathcal{V}_{\mathrm{G}}}y_{ig}},\\
r^{\mathrm{L}}_g(\bm{y})
&=\pos{M_{\mathrm{L}}-\sum_{i\in\mathcal{V}_{\mathrm{L}}}y_{ig}}.
\end{align}
\end{subequations}
The remaining structural residuals are
\begin{subequations}
\label{eq:structural_residuals}
\begin{align}
r^{\mathrm{coh}}_a(\bm{y})
&=\ind\left\{
\substack{\exists i,j\in\mathcal{H}_a,\ g\in\mathcal{K}:\\
y_{ig}\neq y_{jg}}
\right\},\\
r^{\mathrm{sep}}_{ab}(\bm{y})
&=\ind\left\{
\substack{\exists i\in\mathcal H_a,\ j\in\mathcal H_b,\ g\in\mathcal K:\\
y_{ig}y_{jg}=1}
\right\},
\quad 0\le a<b\le K-1.
\end{align}
\end{subequations}
The residuals in \eqref{eq:residuals}, together with
$r^{\mathrm{con}}_g$, vanish exactly when a complete label assignment lies in
$\fset$. Define the positive residual-weight vector
$\bm\rho=(\rho_{\mathrm V},\rho_{\mathrm G},\rho_{\mathrm L},
\rho_{\mathrm{con}},\rho_{\mathrm{coh}},\rho_{\mathrm{sep}})$. The associated
penalized score is
\begin{subequations}
\label{eq:penalized_score}
\begin{align}
\Phi(\bm{y})={}&
\sum_{g\in\mathcal{K}}
\left(
\rho_{\mathrm{V}}r^{\mathrm{V}}_g
+\rho_{\mathrm{G}}r^{\mathrm{G}}_g
+\rho_{\mathrm{L}}r^{\mathrm{L}}_g
+\rho_{\mathrm{con}}r^{\mathrm{con}}_g
\right)\\
&+\rho_{\mathrm{coh}}\sum_{a\in\mathcal{K}}r^{\mathrm{coh}}_a
+\rho_{\mathrm{sep}}
\sum_{0\le a<b\le K-1}r^{\mathrm{sep}}_{ab},\\
S(\bm{y})={}&C(\bm{y})+\Phi(\bm{y}).
\end{align}
\end{subequations}
Here, every residual on the right-hand side is evaluated at $\bm y$. For a
complete label vector, $C(\bm{x})$, $\Phi(\bm{x})$, and $S(\bm{x})$ denote
$C(\bm{y}(\bm{x}))$, $\Phi(\bm{y}(\bm{x}))$, and
$S(\bm{y}(\bm{x}))$, respectively.
The penalty in \eqref{eq:penalized_score} is used only to rank fallback candidates
when a candidate pool contains no feasible assignment. Feasible candidates always
dominate infeasible candidates and, among feasible candidates, selection is made
using the unmodified physical objective $C$.

The regional solver therefore minimizes the physical cut while retaining this
feasibility-first ordering for candidate selection.

%% file: proposed_method.tex
\section{Proposed Sequential Distributed Quantum--Classical Method}
\label{sec:method}

The proposed method limits simultaneous quantum width to a prescribed budget
$Q_{\max}$ by optimizing only a bounded subset of free-bus assignments at each
update. Fixed coherent anchors and a
complete classical assignment provide the boundary information needed to coordinate
the regional quadratic unconstrained binary optimization (QUBO) searches.
Fig.~\ref{fig:workflow} summarizes the resulting workflow, where
$\mathcal V_{\mathrm F}$ denotes the buses outside the coherent anchors, $KN$ the
conventional one-hot assignment width, and $(K-1)|\mathcal V_{\mathrm F}|$ the
symmetry-reduced monolithic width. Only the sampling stage runs on quantum
hardware, and its width is capped by construction, so growth in network size is
absorbed by the number of regional updates and the classical coordination around
them.

\begin{figure*}[t]
\centering
\resizebox{\textwidth}{!}{\input{figures/workflow}}
\caption{Workflow of the proposed method. Preprocessing anchors the coherent
groups, encodes only the free buses, and partitions them into $N_{\mathrm R}$
regions of at most $n_{\max}$ buses, so that every register size
$Q_m\le Q_{\max}$. Each round optimizes the active region $\mathcal R_m$ but
evaluates its candidates on the full network, with projected multipliers
carrying global violations forward without slack qubits. Circled numbers give
the execution order within one update round, crossing the classical and quantum
lanes twice.
Here $p$ is the QAOA depth, $t$ indexes coordination sweeps, unsubscripted $H$
boxes are Hadamard gates, and the cost and mixer Hamiltonians
$\widehat H_{C,m}$, $\widehat H_{M,m}$ are abbreviated $H_C$, $H_M$.}
\label{fig:workflow}
\end{figure*}
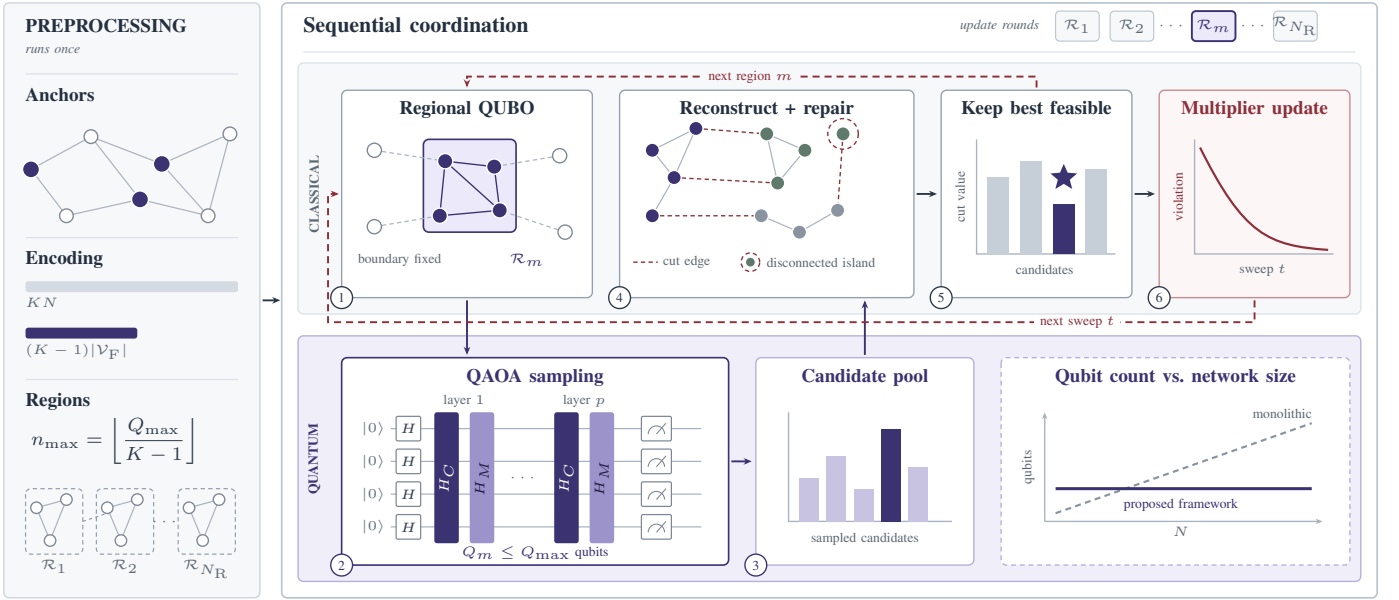

\subsection{Compact Assignment Encoding}

Interchangeable island labels create redundant representations of the same physical
partition, which the proposed method resolves through a compact canonical
representation of assignments~\cite{jiang2026pace}. For $K\ge2$ nonempty,
mutually disjoint coherent groups, the island labels are chosen canonically so
that the memberships of the anchor buses are fixed as
\begin{equation}
y_{ig}=1,\qquad
y_{ih}=0,\quad
\forall g\in\mathcal K,\ \forall i\in\mathcal{H}_g,\
\forall h\in\mathcal K\setminus\{g\}.
\label{eq:anchor_fixing}
\end{equation}
Let
\begin{equation}
\mathcal{V}_{\mathrm{A}}=\bigcup_{g\in\mathcal{K}}\mathcal{H}_g,
\qquad
\mathcal{V}_{\mathrm{F}}=\mathcal{V}\setminus\mathcal{V}_{\mathrm{A}}
\label{eq:anchor_free_sets}
\end{equation}
denote the anchor and free bus sets.

Since the anchor memberships are already fixed by~\eqref{eq:anchor_fixing}, only
the free buses in $\mathcal{V}_{\mathrm{F}}$ require encoding. For each
$i\in\mathcal{V}_{\mathrm{F}}$, the first $K-1$ island labels use the binary
variables
\begin{equation}
z_{ig}\in\{0,1\},\qquad
i\in\mathcal{V}_{\mathrm{F}},\quad g=0,\ldots,K-2.
\label{eq:symmetry_bits}
\end{equation}
The unified membership indicator induced by the symmetry encoding is
\begin{equation}
\widehat y_{ig}(\bm z)=
\begin{cases}
\ind\{g=a\}, & i\in\mathcal{H}_a,\\
z_{ig}, & i\in\mathcal{V}_{\mathrm{F}},\quad 0\le g\le K-2,\\
1-\displaystyle\sum_{h=0}^{K-2}z_{ih},
& i\in\mathcal{V}_{\mathrm{F}},\quad g=K-1.
\end{cases}
\label{eq:symmetry_encoding}
\end{equation}
The all-zero word represents island $K-1$, and the membership indicators sum to
one by construction. A free-bus word is valid when
\begin{equation}
\sum_{g=0}^{K-2}z_{ig}\le 1,
\qquad i\in\mathcal{V}_{\mathrm{F}}.
\label{eq:symmetry_validity}
\end{equation}
For every valid word, $\widehat y_{ig}$ equals the physical membership variable
$y_{ig}$ of the decoded assignment and can therefore be substituted directly
into the model of Section~\ref{sec:formulation}.
The regional QUBO penalizes violations of this condition with
\begin{equation}
\Psi_{\mathrm{enc}}(\bm{z})
=\lambda_{\mathrm{enc}}
\sum_{i\in\mathcal{V}_{\mathrm{F}}}
\sum_{0\le g<h\le K-2}z_{ig}z_{ih}.
\label{eq:symmetry_penalty}
\end{equation}
Here, $\lambda_{\mathrm{enc}}>0$. The penalty is zero for every valid word and
positive whenever two or more explicit bits of a bus are active. Such invalid
measured words are also rejected during decoding.

Because every $\widehat y_{ig}$ is constant or affine in $\bm z$, each linear
membership term remains linear and each pairwise cut term remains at most
quadratic after applying $z_{ig}^{2}=z_{ig}$. The proof that anchor fixing
preserves every physical partition and the complete encoding derivation are given
in~\cite{jiang2026pace}. The assignment-register size needed here is
\begin{equation}
Q_{\mathrm{sym}}=(K-1)|\mathcal{V}_{\mathrm{F}}|
=(K-1)(N-|\mathcal{V}_{\mathrm{A}}|)
\label{eq:symmetry_qubit_count}
\end{equation}
for the symmetry-reduced monolithic representation. For comparison, conventional
one-hot encoding uses $KN$ assignment qubits. The sequential implementation instantiates only
the variables of the active region. Hereafter, the argument $\bm z$ is omitted
from $\widehat y_{ig}$ when clear from context.

Because the number of active variables scales with the number of encoded free
buses, this register size directly determines the admissible regional width.

\subsection{Graph Initialization and Qubit-Bounded Regionalization}

A regional quantum search requires a complete set of external labels and a
decomposition consistent with the qubit limit, shown as the regions panel of
Fig.~\ref{fig:workflow}. A globally complete initial assignment is constructed
by a multi-source shortest-path (MSSP) expansion from the fixed coherent
anchors. Let $w_0=1$ MW be the reference power. The dimensionless length
assigned to branch $\{i,j\}$ is
\begin{equation}
\ell_{ij}=1+\frac{w_0}
{\max(w_{ij},\epsilon_w)},
\label{eq:initialization_length}
\end{equation}
where $\epsilon_w>0$ is a numerical floor expressed in the same power unit as
$w_{ij}$. Each anchor
group starts with distance zero and its associated island label. The label reaching
a free bus with the smallest accumulated length becomes its initial label, with
ties resolved by the smaller island index. The unit term favors topological
proximity, while the inverse-weight term weakly favors paths formed by high-weight
branches.

Because the compact encoding~\cite{jiang2026pace} uses $K-1$ qubits per free bus, let the maximum
simultaneous assignment width satisfy $Q_{\max}\ge K-1$. The resulting regional
bus limit is
\begin{equation}
n_{\max}=\left\lfloor\frac{Q_{\max}}{K-1}\right\rfloor\ge1.
\label{eq:regional_bus_limit}
\end{equation}
When $\mathcal{V}_{\mathrm{F}}$ is nonempty, its buses are ordered by increasing
bus index and divided into the nonempty disjoint regions
$\rset=\{\mathcal{R}_1,\ldots,\mathcal{R}_{N_{\mathrm R}}\}$, where
$N_{\mathrm R}$ is the number of
regions. Thus,
\begin{equation}
\begin{aligned}
\mathcal{V}_{\mathrm{F}}
&=\biguplus_{m=1}^{N_{\mathrm R}}\mathcal{R}_m,
\\
1\le |\mathcal{R}_m|&\le n_{\max},
\qquad m=1,\ldots,N_{\mathrm R},
\\
Q_m&=(K-1)|\mathcal{R}_m|.
\end{aligned}
\label{eq:regional_partition}
\end{equation}
If $\mathcal{V}_{\mathrm{F}}$ is empty, the fixed anchors already determine the
complete assignment and no regional quantum search is required. Otherwise,
\begin{equation}
Q_{\mathrm R}=\max_m Q_m
=(K-1)\max_m|\mathcal{R}_m|
\le Q_{\max},
\label{eq:regional_width}
\end{equation}
independently of $N$. The current regionalizer uses deterministic chunks of the
bus-index ordering. A graph partitioner could replace this regionalizer without
changing the remaining method, but that alternative is not assumed here.

Within a regional update, internal endpoints remain optimization variables whereas
external endpoints are conditioned on the current complete assignment.

\subsection{Boundary-Conditioned Regional QUBO}

Each regional problem must preserve the cut contributions of both internal branches
and branches crossing the regional boundary, as illustrated by the regional QUBO
block of Fig.~\ref{fig:workflow}. At a regional update, let
$\bar{\bm{x}}$ be the latest complete global assignment
and let $\mathcal{R}_m$ be the active region. Define its internal and boundary
edge sets as
\begin{subequations}
\label{eq:regional_edges}
\begin{align}
\mathcal{E}^{\mathrm{int}}_m
&=\{\{i,j\}\in\mathcal{E}:i,j\in\mathcal{R}_m\},\\
\delta(\mathcal{R}_m)
&=\{e\in\mathcal{E}:|e\cap\mathcal{R}_m|=1\}.
\end{align}
\end{subequations}
For a boundary edge, write $i$ for its endpoint in $\mathcal{R}_m$ and $j$ for
the external endpoint. Let $\bm z_m$ collect the variables $z_{ig}$ with
$i\in\mathcal R_m$ and $0\le g\le K-2$. In regional expressions,
$\widehat y_{ig}(\bm z_m)$ denotes the corresponding restriction of
\eqref{eq:symmetry_encoding}. The exact contribution of all edges incident on the
active region, conditioned on $\bar{x}_j$, is
\begin{align}
C_m(\bm{z}_m\mid\bar{\bm{x}})
=\;&
\sum_{\{i,j\}\in\mathcal{E}^{\mathrm{int}}_m}
w_{ij}\left(1-\sum_{g\in\mathcal{K}}
\widehat y_{ig}(\bm z_m)\widehat y_{jg}(\bm z_m)\right)
\nonumber\\
&+
\sum_{\substack{\{i,j\}\in\delta(\mathcal{R}_m)\\i\in\mathcal{R}_m}}
w_{ij}\left(1-\widehat y_{i,\bar{x}_j}(\bm z_m)\right).
\label{eq:regional_cut}
\end{align}
Thus no internal edge is omitted and no boundary edge requires a cross-region
quantum gate. Edges having neither endpoint in the active region are constant
during that update and are excluded.

An inertia term discourages unnecessary changes from the current assignment:
\begin{equation}
C^{\mathrm{in}}_m(\bm{z}_m\mid\bar{\bm{x}})
=\eta\sum_{i\in\mathcal{R}_m}
(1-\widehat y_{i,\bar{x}_i}(\bm z_m)),
\qquad \eta\ge0.
\label{eq:inertia}
\end{equation}
The only encoding-validity term in the active region is the restriction of
\eqref{eq:symmetry_penalty}:
\begin{equation}
\Psi_{\mathrm{enc},m}(\bm z_m)
=\lambda_{\mathrm{enc}}
\sum_{i\in\mathcal{R}_m}
\sum_{0\le g<h\le K-2}z_{ig}z_{ih}.
\label{eq:regional_symmetry_penalty}
\end{equation}

The cut, inertia, and encoding-validity terms are local to the active region.
Island-wide resource counts, however, couple all regions and require global
coordination.

\subsection{Slack-Free Lagrangian Constraint Coordination}

Island-wide resource requirements couple the regional decisions, closing the
multiplier loop of Fig.~\ref{fig:workflow}. Adapting a
Lagrangian multiplier strategy for coupling constraints~\cite{jiang2026pace}, the
sequential coordinator represents these requirements without auxiliary quantum
variables.

Let $\mathcal T=\{\mathrm V,\mathrm G,\mathrm L\}$, with
$\mathcal V_{\mathrm V}=\mathcal V$, and let $M_\theta$ denote the corresponding
minimum in \eqref{eq:cardinality_constraints}. For $\theta\in\mathcal T$ and
$g\in\mathcal K$, define the global count
\begin{equation}
n_g^\theta(\bm x)=\sum_{i\in\mathcal V_\theta}\ind\{x_i=g\}.
\label{eq:global_counts}
\end{equation}
At coordination state $t$, the nonnegative multipliers
$\bm\lambda^t=(\lambda_g^{\theta,t})_{\theta\in\mathcal T,g\in\mathcal K}
\in\mathbb R_+^{3K}$ produce the active-region coefficient
\begin{equation}
\phi_{ig}(\bm\lambda^t)
=-\sum_{\theta\in\mathcal T}\ind\{i\in\mathcal V_\theta\}\lambda_g^{\theta,t}.
\label{eq:lagrangian_reward}
\end{equation}
After assignment-independent terms are omitted, the regional QUBO becomes
\begin{align}
F_m(\bm z_m\mid\bar{\bm x},\bm\lambda^t)
=\;&C_m(\bm z_m\mid\bar{\bm x})
+C^{\mathrm{in}}_m(\bm z_m\mid\bar{\bm x})
\nonumber\\
&+\Psi_{\mathrm{enc},m}(\bm z_m)
\nonumber\\
&+\sum_{i\in\mathcal R_m}\sum_{g\in\mathcal K}
\phi_{ig}(\bm\lambda^t)\widehat y_{ig}(\bm z_m).
\label{eq:regional_qubo}
\end{align}
Sweep $t+1$ produces the complete assignment $\bm x^{t+1}$. The
multipliers are updated by
\begin{equation}
\lambda^{\theta,t+1}_g
=\Pi_{[0,\lambda_{\max}]}
\left(
\lambda^{\theta,t}_g+\alpha_\theta\left[M_\theta-n^\theta_g(\bm{x}^{t+1})\right]
\right),
\label{eq:dual_update}
\end{equation}
where $\theta\in\mathcal T$, $g\in\mathcal K$, $\alpha_\theta>0$,
$\lambda_{\max}>0$, and
$\Pi_{[0,\lambda_{\max}]}(a)=\min\{\lambda_{\max},\max\{0,a\}\}$.
The complete assignment is still tested against $\fset$, so the multipliers guide
the regional search without replacing hard feasibility.

\subsection{Regional QAOA Candidate Generation}

The quantum stage explores alternative assignments within the active region while
the remainder of the network stays fixed. It is the only stage of
Fig.~\ref{fig:workflow} executed on quantum hardware. Quantum sampling
converts the regional
QUBO into a candidate pool. Let $\bm q_m\in\{0,1\}^{Q_m}$ be a fixed vectorization
of the variables in $\bm z_m$. Equation~\eqref{eq:regional_qubo} then has the
standard form
\begin{align}
F_m(\bm q_m\mid\bar{\bm x},\bm\lambda^t)
={}&a_{0,m}+\sum_{\mu=1}^{Q_m}a_{\mu,m}q_{m,\mu}
\nonumber\\
&+\sum_{1\le \mu<\nu\le Q_m}b_{\mu\nu,m}q_{m,\mu}q_{m,\nu}.
\label{eq:qubo_standard}
\end{align}
Here, $a_{0,m}\in\mathbb R$, $a_{\mu,m}\in\mathbb R$, and
$b_{\mu\nu,m}\in\mathbb R$ are the
constant, linear, and quadratic QUBO coefficients, respectively.
The substitution $q_{m,\mu}=(1-Z_\mu)/2$ maps it to a diagonal cost Hamiltonian
$\widehat{H}_{C,m}$. For depth $p$, the regional QAOA state is
\begin{equation}
|\bm{\gamma},\bm{\beta}\rangle_m
=\overleftarrow{\prod}_{k=1}^{p}
e^{-{\rm i}\beta_k\widehat{H}_{M,m}}
e^{-{\rm i}\gamma_k\widehat{H}_{C,m}}
|+\rangle^{\otimes Q_m},
\label{eq:qaoa_state}
\end{equation}
where the arrow places larger $k$ to the left,
$\bm\gamma=(\gamma_1,\ldots,\gamma_p)$,
$\bm\beta=(\beta_1,\ldots,\beta_p)$, and
$\widehat{H}_{M,m}=\sum_{\mu=1}^{Q_m}X_\mu$. Here, $Z_\mu$ and $X_\mu$ denote the Pauli
operators acting on qubit $\mu$. The linear and quadratic QUBO coefficients are
implemented by $R_Z$ and $R_{ZZ}$ rotations, respectively, and the mixer by
$R_X$ rotations.

The $2p$ angles are initialized from a seeded random point and optimized with
constrained optimization by linear approximations (COBYLA), using at most
$I_{\max}$ function evaluations per regional problem.
Expected energies are obtained by exact statevector simulation for smaller
$Q_m$ and by $N_s$-shot matrix-product-state simulation otherwise. After
optimization, distinct valid samples are ordered by measurement frequency. If
sampling produces fewer than the target pool size $B$, classically generated
assignments complete the pool. The current regional assignment is always retained, yielding the hybrid
candidate set $\cset_m$.

Each candidate is embedded into the current complete assignment before global
feasibility, connectivity, and repair are evaluated.

\subsection{Connectivity Evaluation and Repair}

Local encoding validity does not guarantee that the reconstructed islands are
connected, the condition detected in the reconstruct-and-repair block of
Fig.~\ref{fig:workflow}. Strict connectivity is established by the global evaluator after a
complete assignment has been formed. For a reconstructed label vector $\bm x$,
let $\mathcal B_g(\bm x)=\{i\in\mathcal V:x_i=g\}$. An empty
$\mathcal B_g(\bm x)$ is declared disconnected. Otherwise, choose any
$i_g^{(0)}\in\mathcal{B}_g(\bm x)$ and initialize the reached set
$\mathcal{D}^{(0)}_g=\{i_g^{(0)}\}$. The DFS repeatedly applies
\begin{equation}
\mathcal{D}^{(\xi+1)}_g
=\mathcal{D}^{(\xi)}_g
\cup
\left\{
j\in\mathcal{B}_g(\bm x):
\exists i\in\mathcal{D}^{(\xi)}_g,\ \{i,j\}\in\mathcal{E}
\right\}
\label{eq:dfs_reachability}
\end{equation}
until the reached set no longer changes. The nonempty island is connected if and
only if $\mathcal{D}^{(\infty)}_g=\mathcal{B}_g(\bm x)$.
This test is equivalent to requiring every induced island subgraph to be
connected. It uses only the original network edges whose two endpoints belong
to the same island.

Motivated by prior structured postprocessing for physics-constrained
islanding~\cite{jiang2026regrid}, the postprocessor uses connected components
to guide repair. When an island is disconnected, its components are considered in
increasing order of size, and a component is eligible for relocation to another
island only when all of its buses are non-anchor buses. Every repaired
assignment is re-verified by this DFS test, so the
postprocessor only has to propose promising moves while feasibility remains
certified exactly.

Globally evaluated assignments are compared by the feasibility-first ordering
\begin{equation}
\operatorname{rank}(\bm{x})=
\left(
\ind\{\bm{x}\notin\fset\},
S(\bm{x}),
\operatorname{lex}(\bm{x})
\right),
\label{eq:rank}
\end{equation}
with lexicographic tuple comparison. The first entry prioritizes feasibility.
Because $S(\bm x)=C(\bm x)$ for $\bm x\in\fset$, feasible assignments are
ordered by the physical objective, and the bus-index-ordered vector
$\operatorname{lex}(\bm x)$ resolves ties deterministically.

Denote the selected repair map by
$\operatorname{Rep}:\mathcal K^N\rightarrow\mathcal K^N$, with
$\operatorname{Rep}(\bm x)=\bm x$ when repair is disabled. Both the raw
assignment $\bm{x}$ and its repair $\operatorname{Rep}(\bm{x})$ are
globally reevaluated. Using the feasibility-first ordering defined in
\eqref{eq:rank}, the retained result is
\begin{equation}
\operatorname{Safe}(\bm{x})
=\arg\min_{\bm{u}\in\{\bm{x},\operatorname{Rep}(\bm{x})\}}
\operatorname{rank}(\bm{u}),
\label{eq:safe_repair}
\end{equation}
so repair cannot replace a feasible raw assignment by a worse or infeasible
result.

\subsection{Sequential Coordination and Candidate Ranking}

Sequential coordination converts regional samples into globally comparable
controlled-islanding candidates. Within each sweep, the regions are updated in
order, following the two feedback loops of Fig.~\ref{fig:workflow}. The QUBO
for $\mathcal R_m$ uses the latest assignments of
$\mathcal R_1,\ldots,\mathcal R_{m-1}$ and the assignments carried into the sweep
for the remaining regions. Each $\bm v\in\cset_m$ replaces only the active block
and the resulting complete assignment is processed by the safe-repair map
\eqref{eq:safe_repair}.

The working assignment and best-feasible archive are maintained separately.
Here, $T_{\min}$ and $T_{\max}$ are the minimum and maximum numbers of
coordination sweeps, respectively, while $T$ denotes the number actually
completed. The parameter $L$ is the required number of consecutive unchanged
sweeps, and the counter $s$ records this stability condition.
Algorithm~\ref{alg:sequential} gives the complete coordination and stopping rule.

\begin{algorithm}[t]
\caption{Sequential QAOA for Controlled Islanding}
\label{alg:sequential}
\footnotesize
\begin{algorithmic}[1]
\Statex \textbf{Input:} Data of \eqref{eq:full_model} and the parameters
defined above; \textbf{Output:} $\bm{x}^{\mathrm{best}}\in\fset\cup\{\bot\}$,
where $\bot$ denotes no feasible archive
\Statex \textbf{Notation:} $\bm x_{\mathcal R_m}$ restricts $\bm x$ to
$\mathcal R_m$; $\operatorname{MSSP}_{\ell}$ is the initialization in
\eqref{eq:initialization_length};
$\operatorname{Cand}(F_m,p,N_s,B)\subseteq\mathcal K^{|\mathcal R_m|}$ is the
regional candidate generator; $\operatorname{Embed}_m$ replaces block
$\mathcal R_m$ of a complete assignment;
$\operatorname{Best}_{\fset}(\bm a,\bm b)=\bm b$ if $\bm b\in\fset$ and
$[\bm a=\bot\lor C(\bm b)<C(\bm a)]$, and $\bm a$ otherwise
\Statex \textit{Initialization}
\State $\bm{x}\gets\operatorname{MSSP}_{\ell}
(\mathcal G,\mathcal H)$
\State $\bm{x}^{\mathrm{best}}\gets
\operatorname{Best}_{\fset}(\bot,\bm x)$
\If{$\mathcal V_{\mathrm F}=\varnothing$}
    \State \Return $\bm{x}^{\mathrm{best}}$
\EndIf
\State $n_{\max}\gets\lfloor Q_{\max}/(K-1)\rfloor$,
$\rset\gets\{\mathcal R_m\}_{m=1}^{N_{\mathrm R}}$ by
\eqref{eq:regional_partition}
\State $s\gets0$
\Statex \textit{Sequential regional coordination}
\For{$t=0,\ldots,T_{\max}-1$}
    \State $\bm{x}^{\mathrm{old}}\gets\bm{x}$
    \For{$m=1,\ldots,N_{\mathrm R}$}
        \State $\cset_m\gets
        \operatorname{Cand}(F_m(\cdot\mid\bm x,\bm\lambda^t),p,N_s,B)
        \cup\{\bm x_{\mathcal R_m}\}$
        \State $\mathcal U_m\gets
        \{\operatorname{Safe}(\operatorname{Embed}_m(\bm x,\bm v)):
        \bm v\in\cset_m\}$
        \State $\bm x\gets\arg\min_{\bm u\in\mathcal U_m}
        \operatorname{rank}(\bm u)$
        \State $\bm x^{\mathrm{best}}\gets
        \operatorname{Best}_{\fset}(\bm x^{\mathrm{best}},\bm x)$
    \EndFor
    \State $\lambda_g^{\theta,t+1}\gets
    \Pi_{[0,\lambda_{\max}]}(\lambda_g^{\theta,t}
    +\alpha_\theta[M_\theta-n_g^\theta(\bm x)])$,
    $\theta\in\mathcal T$, $g\in\mathcal K$
    \State $s\gets\ind\{\bm x=\bm x^{\mathrm{old}}\}(s+1)$
    \If{$t+1\ge T_{\min}$ and $s\ge L$}
        \State \textbf{break}
    \EndIf
\EndFor
\State \Return $\bm{x}^{\mathrm{best}}$
\end{algorithmic}
\end{algorithm}

\subsection{Properties and Scope of the Method}

The method preserves the best feasible solution encountered, but finite sequential
sampling does not guarantee global optimality.

\begin{proposition}[Archive preservation]
\label{prop:monotonic}
Let $\tau$ index archive updates after the first feasible archive is found, and
let $\tau_0$ be any such update.
Subsequent archive updates satisfy
\begin{align}
C(\bm{x}^{\mathrm{best},\tau+1})
&\le C(\bm{x}^{\mathrm{best},\tau}),
\nonumber\\
C(\bm{x}^{\mathrm{best},\tau_0})=C^\star
&\Longrightarrow
C(\bm{x}^{\mathrm{best},\tau})=C^\star,
\quad \tau\ge\tau_0.
\label{eq:archive_monotonicity}
\end{align}
\end{proposition}

\begin{IEEEproof}
The archive is retained or replaced by a feasible assignment with smaller $C$.
The first relation follows directly. Once $C^\star$ is reached, feasibility gives
$C\ge C^\star$, so equality is preserved.
\end{IEEEproof}

\begin{proposition}[Exhaustive joint recovery]
\label{prop:exhaustive}
If $\cset_m=\mathcal K^{|\mathcal R_m|}$ for every $m$ and the coordinator
evaluates the Cartesian product
$\cset_1\times\cdots\times\cset_{N_{\mathrm R}}$, its best feasible candidate is a global
optimum of \eqref{eq:full_model}.
\end{proposition}

\begin{IEEEproof}
The regions partition $\mathcal V_{\mathrm F}$, and
\eqref{eq:anchor_fixing} fixes the remaining labels without excluding any physical
partition. The Cartesian product therefore contains every feasible assignment.
\end{IEEEproof}

The exhaustive result does not extend to finite sampled pools or sequential block
updates. Even exact regional minimization can reach a coordinatewise local
minimum. Stability therefore terminates the heuristic but does not certify global
optimality.

%% file: figures/workflow.tex
\begingroup%
\definecolor{wfink}{HTML}{243040}%
\definecolor{wfink2}{HTML}{2C3947}%
\definecolor{wfgray}{HTML}{5D6874}%
\definecolor{wfrule}{HTML}{DDE3E9}%
\definecolor{wfedge}{HTML}{AAB4BF}%
\definecolor{wfbord}{HTML}{93A0AD}%
\definecolor{wfbordl}{HTML}{B9C3CD}%
\definecolor{wfpanel}{HTML}{F7F8FA}%
\definecolor{wfpanelb}{HTML}{CCD4DC}%
\definecolor{wfnode}{HTML}{7D8894}%
\definecolor{qpurple}{HTML}{3B3272}%
\definecolor{qlane}{HTML}{F0EEF8}%
\definecolor{qlaneb}{HTML}{B7AFD8}%
\definecolor{qfill}{HTML}{ECEAFA}%
\definecolor{qmix}{HTML}{9B93C9}%
\definecolor{qmixink}{HTML}{231C48}%
\definecolor{qbar}{HTML}{C9C3E2}%
\definecolor{lred}{HTML}{8C2F39}%
\definecolor{lredf}{HTML}{FDF6F7}%
\definecolor{lredb}{HTML}{C98A92}%
\definecolor{gbar}{HTML}{C6CED7}%
\definecolor{gbar2}{HTML}{D5DBE2}%
\definecolor{ggreen}{HTML}{5F7A6D}%
\definecolor{ggray}{HTML}{8A94A0}%
\begin{tikzpicture}[
  x=1pt,y=-1pt,
  every node/.style={inner sep=0pt,outer sep=0pt},
  ttl/.style      ={font=\scriptsize\bfseries,color=wfink},
  btl/.style      ={font=\scriptsize\bfseries,color=wfink},
  lbl/.style      ={font=\tiny,color=wfgray},
  lbli/.style     ={font=\tiny\itshape,color=wfgray},
  lane/.style     ={font=\tiny\bfseries},
  ed/.style       ={draw=wfedge,line width=0.45pt},
  onode/.style    ={draw=wfnode,line width=0.5pt,fill=white},
  cut/.style      ={draw=lred,line width=0.5pt,dash pattern=on 1.6pt off 1.2pt},
  flow/.style     ={draw=wfink2,line width=0.7pt,-{Stealth[length=3pt,width=2.4pt]}},
  qflow/.style    ={draw=qpurple,line width=0.7pt,-{Stealth[length=3pt,width=2.4pt]}},
  loop/.style     ={draw=lred,line width=0.7pt,dash pattern=on 2.6pt off 1.8pt,
                    -{Stealth[length=3pt,width=2.4pt]}},
  axis/.style     ={draw=wfedge,line width=0.5pt},
  stepc/.style    ={circle,draw=wfink2,fill=white,line width=0.5pt,
                    minimum size=8pt,font=\tiny,text=wfink2},
  stepq/.style    ={circle,draw=qpurple,fill=white,line width=0.5pt,
                    minimum size=8pt,font=\tiny,text=qpurple},
]

\useasboundingbox (5.2,13.2) rectangle (508.3,231.3);

\filldraw[rounded corners=2pt,fill=wfpanel,draw=wfpanelb,line width=0.5pt]
  (6,14) rectangle (100,232);
\node[ttl,anchor=north west] at (14,20) {PREPROCESSING};
\node[lbli,anchor=north west] at (14,30) {runs once};
\draw[wfrule,line width=0.5pt] (14,40) -- (92,40);

\node[btl,anchor=north west] at (14,45) {Anchors};
\draw[ed] (16,75) -- (38,63) -- (64,73) -- (89,62);
\draw[ed] (16,75) -- (29,92) -- (56,86) -- (81,92) -- (89,62);
\draw[ed] (38,63) -- (56,86);
\draw[ed] (64,73) -- (81,92);
\foreach \p in {(16,75),(64,73),(56,86)}
  \filldraw[fill=qpurple,draw=white,line width=0.5pt] \p circle (3.1pt);
\foreach \p in {(38,63),(89,62),(29,92),(81,92)}
  \filldraw[onode] \p circle (2.5pt);
\draw[wfrule,line width=0.5pt] (14,100) -- (92,100);

\node[btl,anchor=north west] at (14,105) {Encoding};
\filldraw[rounded corners=1pt,fill=gbar2,draw=none] (14,116) rectangle (92,120);
\node[lbl,anchor=north west] at (14,122) {$KN$};
\filldraw[rounded corners=1pt,fill=qpurple,draw=none] (14,133) rectangle (55,137);
\node[lbl,anchor=north west] at (14,139)
  {$(K-1)\lvert\mathcal V_{\mathrm F}\rvert$};
\draw[wfrule,line width=0.5pt] (14,152) -- (92,152);

\node[btl,anchor=north west] at (14,157) {Regions};
\node[anchor=north west,font=\scriptsize,color=wfink2] at (16,166)
  {$n_{\max}=\left\lfloor\dfrac{Q_{\max}}{K-1}\right\rfloor$};
\foreach \x in {14,40,70}
  \draw[rounded corners=2pt,draw=wfbord,line width=0.45pt,
        dash pattern=on 1.6pt off 1.2pt] (\x,192) rectangle (\x+21,216);
\foreach \x in {14,40,70}{
  \draw[ed] (\x+4,199) -- (\x+15,196) -- (\x+9,211) -- cycle;
  \foreach \p in {(\x+4,199),(\x+15,196),(\x+9,211)}
    \filldraw[onode] \p circle (2.1pt);
}
\draw[ggray,line width=0.4pt,dash pattern=on 1.2pt off 1.2pt]
  (35,204) -- (44,201);
\node[lbl,anchor=north] at (24.5,218) {$\mathcal R_1$};
\node[lbl,anchor=north] at (50.5,218) {$\mathcal R_2$};
\node[lbl,anchor=north] at (80.5,218) {$\mathcal R_{N_{\mathrm R}}$};
\node[lbl] at (65.5,204) {$\cdots$};

\draw[flow] (101,123) -- (107.5,123);

\draw[rounded corners=2pt,draw=wfbordl,line width=0.6pt] (108,14) rectangle (510,232);
\node[anchor=north west,font=\footnotesize\bfseries,color=wfink] at (116,19)
  {Sequential coordination};

\node[lbli,anchor=east] at (386,22.5) {update rounds};
\filldraw[rounded corners=1.5pt,fill=wfpanel,draw=wfbordl,line width=0.4pt]
  (392,17) rectangle (408,28);
\filldraw[rounded corners=1.5pt,fill=wfpanel,draw=wfbordl,line width=0.4pt]
  (412,17) rectangle (428,28);
\filldraw[rounded corners=1.5pt,fill=qfill,draw=qpurple,line width=0.7pt]
  (442,17) rectangle (458,28);
\filldraw[rounded corners=1.5pt,fill=wfpanel,draw=wfbordl,line width=0.4pt]
  (472,17) rectangle (488,28);
\node[lbl] at (400,22.5) {$\mathcal R_1$};
\node[lbl] at (420,22.5) {$\mathcal R_2$};
\node[lbl] at (435,22.5) {$\cdots$};
\node[lbl,color=qpurple] at (450,22.5) {$\mathcal R_m$};
\node[lbl] at (465,22.5) {$\cdots$};
\node[lbl] at (480,22.5) {$\mathcal R_{N_{\mathrm R}}$};
\draw[wfrule,line width=0.5pt] (116,32) -- (502,32);

\filldraw[rounded corners=2pt,fill=wfpanel,draw=wfrule,line width=0.5pt]
  (114,36) rectangle (504,128);
\node[lane,color=wfgray,rotate=90] at (120,84) {CLASSICAL};

\filldraw[rounded corners=2pt,fill=white,draw=wfbord,line width=0.5pt]
  (130,46) rectangle (222,122);
\node[btl,anchor=north] at (176,50) {Regional QUBO};
\filldraw[rounded corners=2pt,fill=qfill,draw=qpurple,line width=0.6pt]
  (160,64) rectangle (194,98);
\draw[ed,dash pattern=on 1.6pt off 1.2pt]
  (142,68) -- (168,72)  (142,96) -- (166,92)
  (210,70) -- (186,74)  (212,98) -- (188,90);
\draw[qpurple,line width=0.5pt]
  (168,72) -- (186,74) -- (188,90) -- (166,92) -- cycle;
\draw[qpurple,line width=0.5pt] (168,72) -- (188,90);
\foreach \p in {(142,68),(142,96),(210,70),(212,98)}
  \filldraw[onode] \p circle (2.6pt);
\foreach \p in {(168,72),(186,74),(188,90),(166,92)}
  \filldraw[fill=qpurple,draw=white,line width=0.5pt] \p circle (2.8pt);
\node[lbl,anchor=north west] at (136,106) {boundary fixed};
\node[lbl,anchor=north,color=qpurple] at (198,106) {$\mathcal R_m$};

\filldraw[rounded corners=2pt,fill=white,draw=wfbord,line width=0.5pt]
  (232,46) rectangle (340,122);
\node[btl,anchor=north] at (286,50) {Reconstruct + repair};
\draw[ed] (244,68) -- (260,60) -- (252,78) -- cycle;
\draw[ed] (252,78) -- (244,92);
\draw[ed] (286,62) -- (300,68) -- (290,80) -- cycle;
\draw[ed] (284,92) -- (298,98) -- (312,90);
\draw[cut] (260,60) -- (286,62)  (252,78) -- (290,80)
           (244,92) -- (284,92)  (314,62) -- (312,90);
\foreach \p in {(244,68),(260,60),(252,78),(244,92)}
  \filldraw[fill=qpurple,draw=white,line width=0.5pt] \p circle (2.6pt);
\foreach \p in {(286,62),(300,68),(290,80),(314,62)}
  \filldraw[fill=ggreen,draw=white,line width=0.5pt] \p circle (2.6pt);
\foreach \p in {(284,92),(298,98),(312,90)}
  \filldraw[fill=ggray,draw=white,line width=0.5pt] \p circle (2.6pt);
\draw[cut] (314,62) circle (5.6pt);
\draw[cut] (237,109) -- (246,109);
\node[lbl,anchor=west] at (248,109) {cut edge};
\filldraw[fill=ggreen,draw=none] (280,109) circle (1.6pt);
\draw[cut] (280,109) circle (3.4pt);
\node[lbl,anchor=west] at (286,109) {disconnected island};

\filldraw[rounded corners=2pt,fill=white,draw=wfbord,line width=0.5pt]
  (350,46) rectangle (420,122);
\node[btl,anchor=north] at (385,50) {Keep best feasible};
\draw[axis] (363,64) -- (363,106) -- (414,106);
\filldraw[fill=gbar,draw=none] (367,78) rectangle (375,106);
\filldraw[fill=gbar,draw=none] (379,72) rectangle (387,106);
\filldraw[fill=qpurple,draw=none] (391,88) rectangle (399,106);
\filldraw[fill=gbar,draw=none] (403,75) rectangle (411,106);
\filldraw[fill=qpurple,draw=none]
  (395,73) -- (396.4,76.2) -- (399.9,76.5) -- (397.2,78.8) -- (398.1,82.2)
  -- (395,80.3) -- (391.9,82.2) -- (392.8,78.8) -- (390.1,76.5)
  -- (393.6,76.2) -- cycle;
\node[lbl,anchor=north] at (388,110) {candidates};
\node[lbl,rotate=90] at (357,85) {cut value};

\filldraw[rounded corners=2pt,fill=lredf,draw=lredb,line width=0.5pt]
  (430,46) rectangle (500,122);
\node[btl,anchor=north,color=lred] at (465,50) {Multiplier update};
\draw[axis] (443,64) -- (443,106) -- (494,106);
\draw[lred,line width=0.8pt]
  (445,67) .. controls (461,99) and (470,103) .. (492,104.5);
\node[lbl,anchor=north] at (468,110) {sweep $t$};
\node[lbl,rotate=90,color=lred] at (437,85) {violation};

\draw[flow] (341,84) -- (348.5,84);
\draw[flow] (421,84) -- (428.5,84);

\filldraw[rounded corners=2pt,fill=qlane,draw=qlaneb,line width=0.5pt]
  (114,136) rectangle (504,226);
\node[lane,color=qpurple,rotate=90] at (120,181) {QUANTUM};

\filldraw[rounded corners=2pt,fill=white,draw=qpurple,line width=0.6pt]
  (130,144) rectangle (272,220);
\node[btl,anchor=north,color=qpurple] at (201,148) {QAOA sampling};
\foreach \y in {170,182,194,206}{
  \draw[ed] (148,\y) -- (262,\y);
  \node[lbl,anchor=east] at (146,\y) {$\lvert 0\rangle$};
  \filldraw[fill=white,draw=wfnode,line width=0.45pt,rounded corners=0.6pt]
    (150,\y-4.5) rectangle (159,\y+4.5);
  \node[font=\tiny,color=wfink2] at (154.5,\y) {$H$};
  \filldraw[fill=white,draw=wfnode,line width=0.45pt,rounded corners=0.6pt]
    (240,\y-4.5) rectangle (251,\y+4.5);
  \draw[wfnode,line width=0.4pt] (242.5,\y+2) arc (180:0:3pt);
  \draw[wfnode,line width=0.4pt] (245.5,\y+2) -- (249,\y-2);
}
\filldraw[fill=qpurple,draw=none,rounded corners=1pt] (164,164) rectangle (173,212);
\node[font=\tiny,color=white,rotate=90] at (168.5,188) {$H_C$};
\filldraw[fill=qmix,draw=none,rounded corners=1pt] (177,164) rectangle (186,212);
\node[font=\tiny,color=qmixink,rotate=90] at (181.5,188) {$H_M$};
\filldraw[fill=qpurple,draw=none,rounded corners=1pt] (208,164) rectangle (217,212);
\node[font=\tiny,color=white,rotate=90] at (212.5,188) {$H_C$};
\filldraw[fill=qmix,draw=none,rounded corners=1pt] (221,164) rectangle (230,212);
\node[font=\tiny,color=qmixink,rotate=90] at (225.5,188) {$H_M$};
\node[lbl] at (197,188) {$\cdots$};
\node[lbl,anchor=south] at (175,162) {layer $1$};
\node[lbl,anchor=south] at (219,162) {layer $p$};
\node[lbl,anchor=north,color=qpurple] at (201,213)
  {$Q_m\le Q_{\max}$ qubits};

\filldraw[rounded corners=2pt,fill=white,draw=qlaneb,line width=0.5pt]
  (282,144) rectangle (362,220);
\node[btl,anchor=north,color=qpurple] at (322,148) {Candidate pool};
\draw[axis] (294,164) -- (294,204) -- (354,204);
\filldraw[fill=qbar,draw=none] (298,188) rectangle (305,204);
\filldraw[fill=qbar,draw=none] (308,180) rectangle (315,204);
\filldraw[fill=qbar,draw=none] (318,192) rectangle (325,204);
\filldraw[fill=qpurple,draw=none] (328,170) rectangle (335,204);
\filldraw[fill=qbar,draw=none] (338,184) rectangle (345,204);
\node[lbl,anchor=north] at (322,208) {sampled candidates};

\draw[qflow] (273,182) -- (280.5,182);

\filldraw[rounded corners=2pt,fill=white,draw=qlaneb,line width=0.5pt,
          dash pattern=on 1.8pt off 1.4pt]
  (372,144) rectangle (500,220);
\node[btl,anchor=north,color=qpurple] at (436,148)
  {Qubit count vs.\ network size};
\draw[axis,-{Stealth[length=2.6pt,width=2pt]}] (388,204) -- (388,164);
\draw[axis,-{Stealth[length=2.6pt,width=2pt]}] (388,204) -- (490,204);
\draw[ggray,line width=0.8pt,dash pattern=on 2.2pt off 1.6pt]
  (392,201) -- (486,168);
\node[lbl,anchor=south east] at (486,166) {monolithic};
\draw[qpurple,line width=1.1pt] (392,192) -- (486,192);
\node[lbl,anchor=north,text=qpurple] at (438,195.5) {proposed framework};
\node[lbl,anchor=north] at (438,206) {$N$};
\node[lbl,rotate=90] at (382,184) {qubits};

\draw[qflow] (176,123) -- (176,143);
\draw[qflow] (322,143) -- (322,123);

\node[stepc] at (130,122) {1};
\node[stepq] at (130,220) {2};
\node[stepq] at (282,220) {3};
\node[stepc] at (232,122) {4};
\node[stepc] at (350,122) {5};
\node[stepc] at (430,122) {6};

\draw[loop] (385,45) -- (385,41) -- (176,41) -- (176,45);
\node[lbl,fill=wfpanel,inner xsep=2pt,inner ysep=0.5pt,text=lred] at (280,41)
  {next region $m$};

\draw[loop] (465,123) -- (465,131) -- (125,131) -- (125,84) -- (129,84);
\node[lbl,fill=white,inner xsep=2pt,inner ysep=0.5pt,text=lred] at (400,131)
  {next sweep $t$};

\end{tikzpicture}%
\endgroup%

%% file: numerical_examples.tex
\section{Numerical Examples}
\label{sec:numerical_examples}

In this section, we present the reconstructed islanding solutions, compare cut
values across sampling platforms and Gurobi reference solutions, and examine
the resource scaling and hardware execution cost of the proposed method. The
method is evaluated on eleven IEEE benchmark systems ranging from $9$ to $300$
buses. All branch weights and cut values are reported in MW.
Table~\ref{tab:experiment_setup} summarizes the QAOA and coordination settings,
including the island count $K$, depth $p$, shot count $N_s$, candidate-pool
size $B$, regional bus and qubit limits, and stopping parameters. The main experiments are
executed on the IBM Quantum \texttt{ibm\_fez} processor. To assess noise
resilience, the hardware results are compared with corresponding runs on an
ideal Qiskit Aer statevector simulator and on an Aer simulator equipped with a
noise model calibrated from an IBM quantum processor.

\begin{table}[t]
\centering
\caption{Test-system execution settings.}
\label{tab:experiment_setup}
\setlength{\tabcolsep}{1.7pt}
\resizebox{0.90\columnwidth}{!}{%
\begin{tabular}{lrrrrrrrrrr}
\toprule
System & $K$ & $p$ & $N_s$ & $B$ & $n_{\max}$ & $Q_{\max}$ &
$T_{\min}/T_{\max}$ & $L$ & $I_{\max}$ & $T$ \\
\midrule
$9,14,24$   & $2,2,3$ & 1 & 100  & 8  & $3,9,4$ & $3,9,8$ & 1/5   & 2 & 20 & 3 \\
$30$        & 2       & 1 & 300  & 16 & 10      & 10      & 5/15  & 3 & 25 & 10 \\
$39$        & 3       & 1 & 1000 & 16 & 4       & 8       & 5/15  & 3 & 25 & 10 \\
$57$        & 2       & 2 & 1500 & 8  & 10      & 10      & 5/15  & 3 & 8  & 10 \\
$73$        & 3       & 3 & 2000 & 8  & 6       & 12      & 5/15  & 3 & 8  & 10 \\
$89$        & 3       & 3 & 3000 & 8  & 6       & 12      & 8/20  & 3 & 8  & 8 \\
$118,145$   & 4       & 3 & 3000 & 8  & 6       & 18      & 8/20  & 3 & 8  & 8 \\
$300$       & 3       & 5 & 5000 & 8  & 6       & 12      & 8/20  & 3 & 10 & 8 \\
\bottomrule
\end{tabular}
}
\end{table}

\subsection{Representative Islanding Solutions}

In this section, the reconstructed islanding results for the IEEE benchmark
systems are illustrated and compared with the Gurobi reference values.

Table~\ref{tab:island_solutions} reports the reconstructed feasible partitions
for the IEEE 9- through 73-bus systems. The complete bus-to-island assignments
document the island memberships underlying the cut values compared in
Table~\ref{tab:platform_quality}. For the 9-, 14-, 24-, 30-, 39-, 57-, and
73-bus systems, the corresponding cut values are $1.24$, $88.24$, $776.21$,
$16.86$, $228.99$, $128.24$, and $284.54$ MW, respectively.
For presentation, the canonical labels $0,\ldots,K-1$ are renumbered as
Islands 1 through $K$ in the table and figures.

Figs.~\ref{fig:islanding_small}--\ref{fig:islanding_300} then show the
reconstructed feasible islanding solutions for the four larger benchmark
systems. Node color denotes the assigned island, whereas dashed edges identify
the transmission branches opened between islands. Colors distinguish the
islands but are not tied to canonical numerical labels. Thus, the figures provide a
network-level view of the assignments returned by the regional optimization and
show that every colored subgraph remains connected after the cut. The coherent
generator groups are placed in separate islands in all four cases.

The IEEE 89-bus solution contains three islands with $38$, $31$, and $20$ buses
and has a cut value of $3790.94$ MW. The IEEE 118-bus solution contains four
islands with $29$, $19$, $24$, and $46$ buses and has a cut value of $634.05$
MW.
For the IEEE 145-bus system, the reconstructed partition comprises one
$124$-bus island and three $7$-bus islands, with a cut value of $11859.80$ MW.
The IEEE 300-bus solution forms three connected islands of $66$, $143$, and
$91$ buses, respectively, and attains a cut value of $2322.40$ MW. The separate
145- and 300-bus panels preserve the readability of individual bus labels and
make the inter-island boundaries visible in the two densest networks.

Table~\ref{tab:platform_quality} compares these solutions with the Gurobi
reference values. For every reported benchmark, the best cut value obtained by
the proposed method on the IBM quantum processor matches the corresponding
Gurobi optimum while satisfying all islanding constraints. This agreement shows
that, for the tested configurations, the sequential distributed method can
recover optimal islanding solutions using real quantum hardware. It is empirical
evidence for the reported cases rather than a general convergence guarantee.

\begin{table}[t]
\centering
\caption{Islanding solutions by bus number for the test systems.}
\label{tab:island_solutions}
\scriptsize
\setlength{\tabcolsep}{4pt}
\renewcommand{\arraystretch}{1.15}
\begin{tabular}{@{}l@{\hspace{5pt}}p{0.78\columnwidth}@{}}
\toprule
System & Islands (bus numbers) \\
\midrule
9-bus  & Island 1: 2--3, 6--9\quad Island 2: 1, 4--5 \\
14-bus & Island 1: 6--14\quad Island 2: 1--5 \\
24-bus & Island 1: 3, 15--19, 21--22, 24\quad Island 2: 6, 10--14, 20, 23\quad
         Island 3: 1--2, 4--5, 7--9 \\
30-bus & Island 1: 10, 12--30\quad Island 2: 1--9, 11 \\
39-bus & Island 1: 4--14, 31--32, 39\quad Island 2: 15--24, 33--36\quad
         Island 3: 1--3, 25--30, 37--38 \\
57-bus & Island 1: 1--5, 11, 13--23, 32--49, 56--57\quad Island 2: 6--10, 12,
         24--31, 50--55 \\
73-bus & Island 1: 7--8, 25--48\quad Island 2: 49--72\quad Island 3: 1--6, 9--24,
         73 \\
\bottomrule
\end{tabular}
\end{table}

Table~\ref{tab:platform_quality} uses $C_{\mathrm I}$, $C_{\mathrm N}$, and
$C_{\mathrm H}$ for the best cut weights obtained from ideal Aer,
IBM-calibrated noisy Aer, and IBM hardware, respectively, while
$C^\star$ denotes the Gurobi optimum. The $G$ and $d$ triplets report
the transpiled two-qubit gate count and overall circuit depth in the same
order.

\begin{figure}[t]
\centering
\includegraphics[trim=0 18pt 0 12pt,clip,width=0.92\columnwidth]{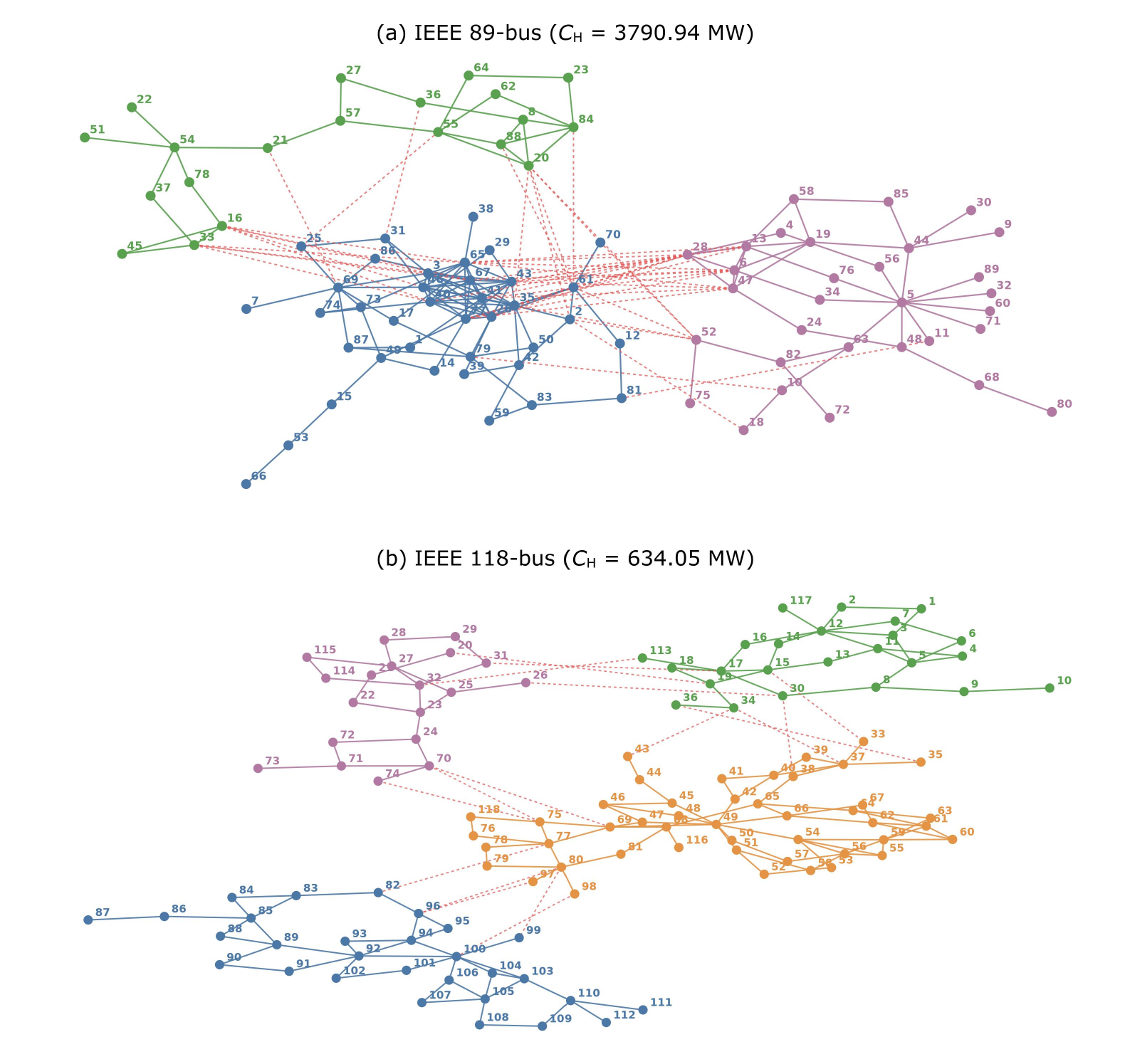}
\caption{Reconstructed feasible islanding solutions on the IEEE 89- and
118-bus systems. Node color denotes island membership, dashed edges mark
inter-island cuts, and each panel reports the corresponding hardware cut value.}
\label{fig:islanding_small}
\end{figure}

\begin{figure}[t]
\centering
\includegraphics[trim=0 25pt 0 12pt,clip,width=0.92\columnwidth]{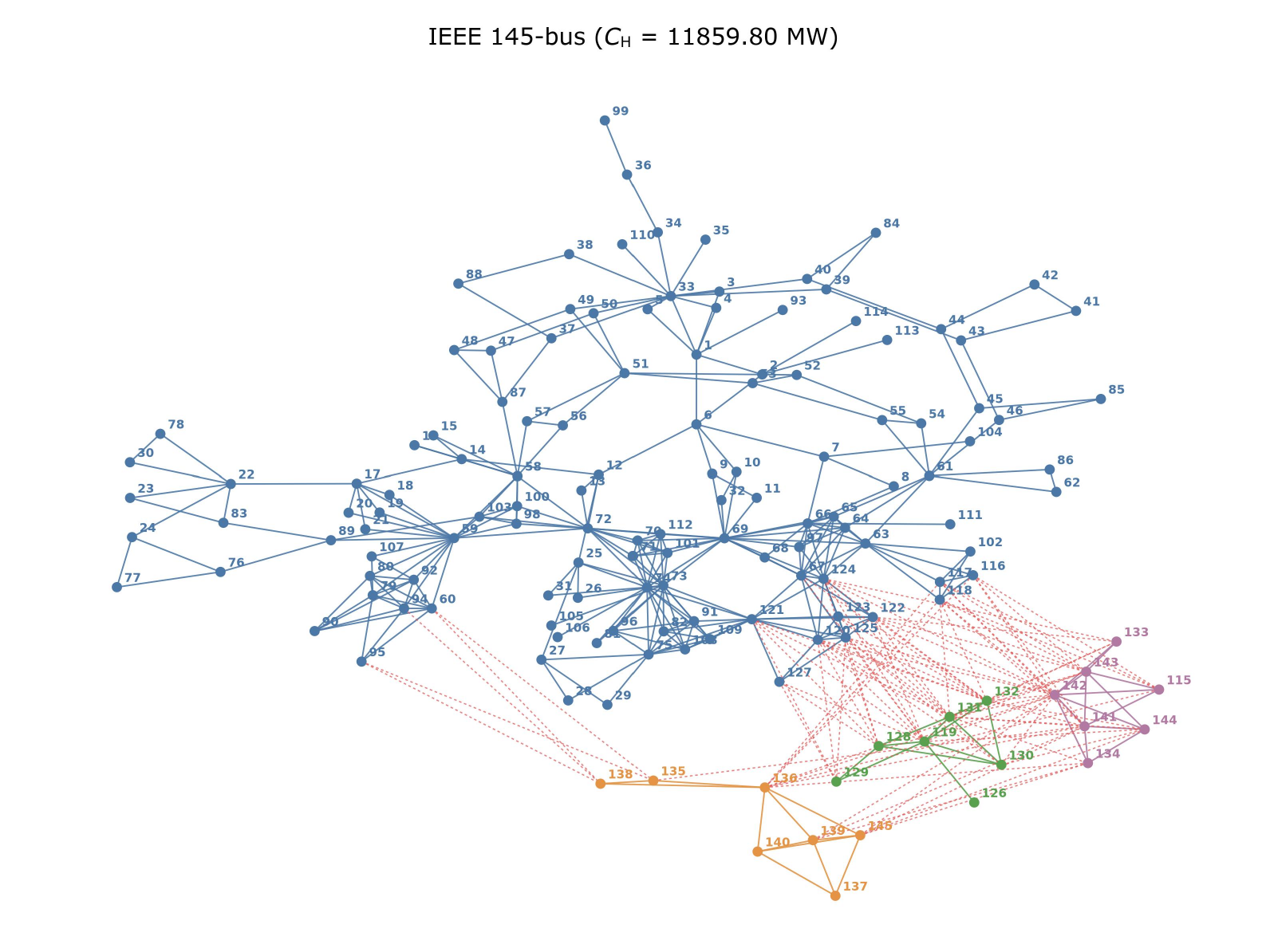}
\caption{Reconstructed feasible islanding solution on the IEEE 145-bus system.}
\label{fig:islanding_145}
\end{figure}

\begin{figure}[t]
\centering
\includegraphics[trim=0 48pt 0 12pt,clip,width=\columnwidth]{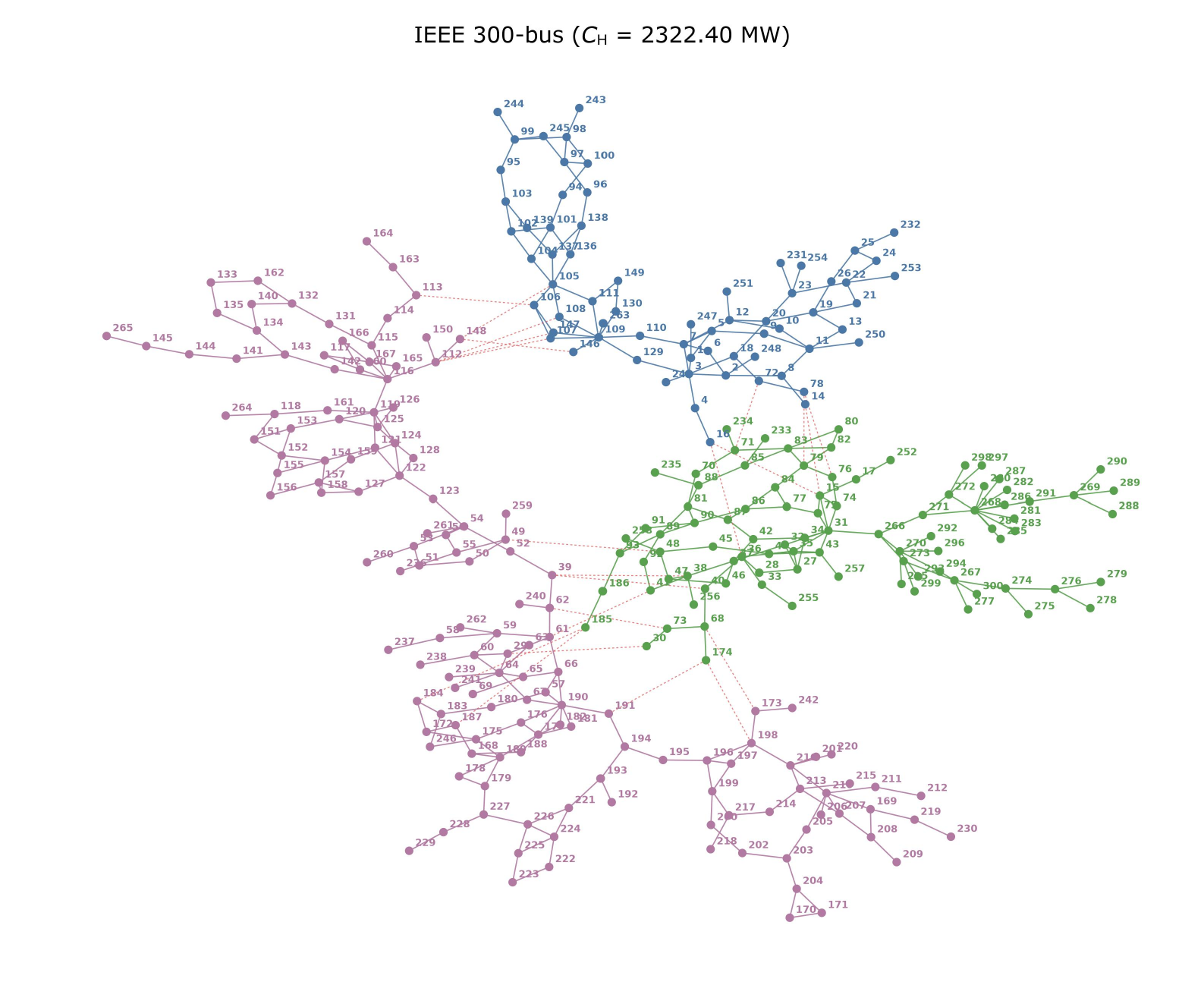}
\caption{Reconstructed feasible islanding solution on the IEEE 300-bus system.}
\label{fig:islanding_300}
\end{figure}

\begin{table}[t]
\centering
\caption{Cross-platform cut values (MW), two-qubit gate counts $G$, and circuit
depths $d$.}
\label{tab:platform_quality}
\scriptsize
\setlength{\tabcolsep}{1.7pt}
\begin{tabular}{lrrrrrr}
\toprule
System & $C_{\mathrm I}$ & $C_{\mathrm N}$ & $C_{\mathrm H}$ &
$C^\star$ & $G$ & $d$ \\
\midrule
9-bus   & 1.24     & 1.24     & 1.24     & 1.24     & 2/4/4             & 6/22/22 \\
14-bus  & 88.24    & 88.24    & 88.24    & 88.24    & 9/22/22           & 11/67/71 \\
24-bus  & 776.21   & 776.21   & 776.21   & 776.21   & 6/16/16           & 5/36/36 \\
30-bus  & 16.86    & 16.86    & 16.86    & 16.86    & 6/12/12           & 7/39/37 \\
39-bus  & 228.99   & 228.99   & 228.99   & 228.99   & 16/55/56          & 10/112/119 \\
57-bus  & 128.24   & 128.24   & 128.24   & 128.24   & 14/26/26          & 11/78/83 \\
73-bus  & 284.54   & 284.54   & 284.54   & 284.54   & 30/73/75          & 16/195/203 \\
89-bus  & 3790.94  & 3790.94  & 3790.94  & 3790.94  & 24/49/57          & 11/124/150 \\
118-bus & 634.05   & 634.05   & 634.05   & 634.05   & 108/433/435      & 28/221/229 \\
145-bus & 11859.80 & 11859.80 & 11859.80 & 11859.80 & 135/615/616      & 46/549/566 \\
300-bus & 2322.40  & 2322.40  & 2322.40  & 2322.40  & 62/220/220        & 35/506/260 \\
\bottomrule
\end{tabular}
\end{table}

\subsection{Hardware Execution Cost}

Table~\ref{tab:hardware_cost} uses $N_{\mathrm R}$ for the region count,
$N_c=N_{\mathrm R}T$ for the number of executed regional circuits, and
$Q_{\mathrm R}=\max_m Q_m$ for the largest simultaneous regional register.
The symmetry-encoded monolithic register width is $Q_{\mathrm{sym}}$ from
\eqref{eq:symmetry_qubit_count}. The factor
$\chi=KN/Q_{\mathrm R}$ measures
the width reduction relative to one-hot
encoding and is reported to the nearest integer. The quantities $G$, $d$, and
$t_q$ denote the transpiled two-qubit gate count, overall circuit depth, and
accumulated quantum processing unit (QPU) time, respectively.
For the proposed method, $G$, $d$, and $t_q$ are obtained from the IBM hardware
executions. Within each column group, the quantities refer to the method identified
by the group heading. An N/A entry indicates that the monolithic circuit exceeds
the available hardware width.

Regionalization keeps the simultaneous circuit width between $3$ and $18$
qubits as the network grows from $9$ to $300$ buses. By comparison,
the monolithic qubit count increases from $6$ to $462$. The difference is
already $12$ versus $154$ qubits for the 89-bus system and reaches $18$ versus
$285$ at 145 buses and $12$ versus $462$ at 300 buses. The monolithic
requirements of $192$, $285$, and $462$ qubits for the three largest systems
exceed the $156$ physical qubits of \texttt{ibm\_fez}. The proposed method
accommodates this growth through additional regions and circuit executions,
with $N_{\mathrm R}$ increasing to $39$ and $N_c$ to $312$, while preserving
a bounded simultaneous register. This structure yields a one-hot width
reduction of up to $75\times$.

The bounded register also reduces the compiled burden of each circuit. For all
eight systems executed by both methods, the proposed method has lower $G$ and
$d$. At 24 buses, the two-qubit gate count decreases from $248$ to $16$ and the
depth from $94$ to $36$. At 39 buses, the corresponding reductions are from
$2227$ to $56$ and from $720$ to $119$. The advantage persists at 89 buses,
where the proposed circuit contains $57$ two-qubit gates at depth $150$, while
the monolithic circuit contains $3129$ two-qubit gates at depth $865$. Thus,
the decomposition controls both circuit width and per-circuit compilation
overhead as the problem size increases.

Sequential regional execution enables the proposed method to retain a bounded
hardware footprint as the network size increases. The reported $t_q$ therefore
accumulates the processor time of the regional circuits. From 9 to 57 buses,
the proposed $t_q$ ranges from $4$ to $23$~s,
compared with $2$ to $6$~s for the monolithic baseline. At 73 and 89 buses, the
corresponding values are $50$ versus $26$~s and $101$ versus $47$~s. However,
the increase remains moderate in absolute terms. Even the 300-bus system
requires only $263$~s of QPU time while retaining a 12-qubit
simultaneous register. The regional formulation also remains executable for
the 118- and 145-bus systems with QPU times of $146$ and $176$~s, whereas the
corresponding monolithic circuits exceed the available qubit capacity. The
results therefore show a favorable scalability trade-off in which manageable
sequential execution supports bounded circuit width and substantially lower
 circuit complexity.

\begin{table}[t]
\centering
\caption{Hardware cost and monolithic-QAOA resources.}
\label{tab:hardware_cost}
\scriptsize
\setlength{\tabcolsep}{1.5pt}
\resizebox{0.96\columnwidth}{!}{%
\begin{tabular}{lrrrrrrrrrrr}
\toprule
& \multicolumn{6}{c}{Proposed method} & \multicolumn{1}{c}{Width reduction} &
\multicolumn{4}{c}{Monolithic QAOA} \\
\cmidrule(lr){2-7}\cmidrule(lr){8-8}\cmidrule(l){9-12}
System & $N_{\mathrm R}$ & $N_c$ & $Q_{\mathrm R}$ & $G$ & $d$ &
$t_q$ (s) & $\chi$ & $Q_{\mathrm{sym}}$ & $G$ & $d$ & $t_q$ (s) \\
\midrule
9-bus   & 2  & 6   & 3  & 4   & 22  & 4   & $6\times$  & 6   & 21   & 25  & 2 \\
14-bus  & 1  & 3   & 9  & 22  & 71  & 6   & $3\times$  & 9   & 27   & 77  & 2 \\
24-bus  & 4  & 12  & 8  & 16  & 36  & 6   & $9\times$  & 26  & 248  & 94  & 3 \\
30-bus  & 3  & 30  & 10 & 12  & 37  & 10  & $6\times$  & 24  & 127  & 62  & 4 \\
39-bus  & 8  & 80  & 8  & 56  & 119 & 16  & $15\times$ & 58  & 2227 & 720 & 6 \\
57-bus  & 5  & 50  & 10 & 26  & 83  & 23  & $11\times$ & 50  & 638  & 544 & 6 \\
73-bus  & 5  & 50  & 12 & 75  & 203 & 50  & $18\times$ & 54  & 2575 & 719 & 26 \\
89-bus  & 13 & 104 & 12 & 57  & 150 & 101 & $22\times$ & 154 & 3129 & 865 & 47 \\
118-bus & 11 & 88  & 18 & 435 & 229 & 146 & $26\times$ & 192 & \multicolumn{3}{c}{N/A} \\
145-bus & 16 & 128 & 18 & 616 & 566 & 176 & $32\times$ & 285 & \multicolumn{3}{c}{N/A} \\
300-bus & 39 & 312 & 12 & 220 & 260 & 263 & $75\times$ & 462 & \multicolumn{3}{c}{N/A} \\
\bottomrule
\end{tabular}
}
\end{table}

\subsection{Solution Quality and Noise Resilience}

In this section, the noise resilience of the proposed method is evaluated using
the solution-quality and circuit-resource data in
Table~\ref{tab:platform_quality}. The comparison covers ideal Aer simulation,
IBM-calibrated noisy simulator, and IBM hardware under the common notation
and platform order defined with the table.

The identical cut values obtained from ideal simulator, calibrated-noise
simulator, IBM hardware, and Gurobi across all test systems demonstrate
the robustness of the proposed method to modeled and physical quantum noise.
In both noise-affected settings, the method retains the Gurobi-optimal cut
obtained under ideal simulation while producing feasible reconstructed
partitions. This consistency indicates that the regional decomposition and
classical reconstruction limit the influence of the tested noise on the final
islanding decision.

The circuit data further characterize the effect of hardware-aware compilation
on the quantum circuits that yield these optimal solutions. Ideal Aer circuits
generally have smaller $G$ and $d$ because they require no hardware routing. In
contrast, the noisy-simulator
and hardware two-qubit gate counts are identical for six systems and differ by
no more than eight gates in every other case. Their depths differ by at most
$26$ through the 145-bus system. For the 300-bus system, both compiled circuits
contain $220$ two-qubit gates, although their depths are $506$ and $260$. This
difference indicates that gate scheduling can change the overall depth without
increasing the two-qubit gate volume. Despite these compilation variations, the
cut values and feasibility remain unchanged, providing consistent evidence of
noise resilience across all reported systems.

\subsection{Computational Complexity}

By \eqref{eq:regional_width}, every regional register satisfies $Q_m\le Q_{\max}$,
independently of $N$. The following scaling analysis treats the island count
$K$ as fixed. On networks of bounded degree, each region contributes
$O(Q_{\max})$ quadratic terms to \eqref{eq:regional_qubo}, so a single regional
circuit of depth $p$ uses $G_{\mathrm{reg}}=O(pQ_{\max})$ two-qubit gates and
$2p$ variational parameters, both independent of $N$. Partitioning
$\mathcal V_{\mathrm F}$ into regions of size at most
$n_{\max}=\lfloor Q_{\max}/(K-1)\rfloor$ requires
$N_{\mathrm R}=O(N/Q_{\max})$ regions, so
the aggregate quantum workload per sweep is
\begin{equation}
G_{\mathrm{sweep}}=N_{\mathrm R}\,G_{\mathrm{reg}}=O(pN),
\label{eq:gate_count_sweep}
\end{equation}
linear in network size even though no individual circuit exceeds $Q_{\max}$
qubits. This replaces the monolithic state space
$|\mathcal S_{\mathrm{mono}}|=O(2^{Q_{\mathrm{sym}}})$, where
$\mathcal S_{\mathrm{mono}}$ denotes the monolithic assignment state space. The regional formulation instead uses $N_{\mathrm R}$ subproblems of
fixed size $Q_{\max}$ as $N$ grows.

On the classical side, every candidate in $\cset_m$ is checked by the DFS
connectivity test \eqref{eq:dfs_reachability} on the full network graph, costing
$O(N+|\mathcal E|)$. With pool size $B$, let
$W_{\mathrm{reg}}^{\mathrm{cl}}=O(B(N+|\mathcal E|))$ denote the per-region
classical work. The per-sweep classical work is
\begin{equation}
W_{\mathrm{sweep}}^{\mathrm{cl}}=N_{\mathrm R}\,W_{\mathrm{reg}}^{\mathrm{cl}}
=O\!\left(\frac{BN(N+|\mathcal E|)}{Q_{\max}}\right).
\label{eq:classical_cost_sweep}
\end{equation}
For sparse power networks, $|\mathcal E|=O(N)$, so
$W_{\mathrm{sweep}}^{\mathrm{cl}}=O(BN^{2}/Q_{\max})$,
and global feasibility evaluation, rather than the bounded-width quantum stage,
becomes the dominant cost as $N$ grows. If $T\le T_{\max}$ is the realized
number of sweeps, the total
resource use is
\begin{equation}
G_{\mathrm{tot}}=O(TpN),
\qquad
W_{\mathrm{tot}}^{\mathrm{cl}}=O\!\left(\frac{TBN^2}{Q_{\max}}\right),
\label{eq:total_complexity}
\end{equation}
while the quantum circuit width itself remains fixed at $Q_{\max}$ throughout.

The measured resources in Table~\ref{tab:hardware_cost} support this scaling
analysis. The largest simultaneous register $Q_{\mathrm R}$ does not exceed
$18$ qubits, while
the symmetry-encoded monolithic baseline reaches $462$ qubits and the one-hot encoding
requires $900$ qubits for the 300-bus system. Growth is absorbed by the region
count, which rises from $2$ to $39$ in agreement with the
$N_{\mathrm R}=O(N/Q_{\max})$ scaling derived above, yielding a $75\times$ width
reduction over one-hot encoding.

The exponential state space of a monolithic formulation eventually overwhelms
any fixed quantum device, capping the largest network it can ever address.
Holding every regional circuit at $Q_{\max}$ qubits by construction instead
absorbs growth in $N$ into the number of subproblems and the classical
coordination around them, so the device that solves a small network remains
sufficient for an arbitrarily large one. Decoupling the quantum resource
requirement from network size makes near-term, qubit-limited hardware viable for
networks no monolithic formulation could ever fit.

%% file: conclusion.tex
\section{Conclusion}
\label{sec:conclusion}

This paper developed a qubit-bounded sequential distributed QAOA framework for
coherent controlled islanding under limited quantum resources. It decomposes
free-bus assignments into boundary-conditioned regional QUBOs that are
optimized sequentially within a fixed qubit budget. Across eleven IEEE systems
ranging from 9 to 300 buses, the method recovered feasible partitions with
Gurobi-optimal cut values on all three sampling platforms. The identical
solution quality under ideal simulation, calibrated-noise simulation, and IBM
hardware execution demonstrates resilience to both modeled and physical
quantum noise. The regional formulation reduced the simultaneous register
width, two-qubit gate count, and circuit depth relative to monolithic QAOA
across all common test systems.
These findings demonstrate that large controlled-islanding
problems can be mapped to current quantum hardware without increasing the
quantum register with network size or degrading the reported solution quality
under the tested noise conditions. More broadly, the framework provides a
scalable pathway for advancing quantum applications in power-system
optimization.

Future work will investigate adaptive regionalization and parallel regional
execution to improve finite-sampling convergence and reduce QPU time, and
extend the framework to other quantum platforms.